\documentclass[sn-mathphys-num]{sn-jnl}

\usepackage{graphicx}%
\usepackage{amsmath,amssymb,amsfonts}%
\usepackage[scr=rsfs]{mathalpha}
\usepackage[title]{appendix}%
\usepackage{xurl}%
\usepackage{textcomp}%
\usepackage{listings}%
\usepackage{tabularx}
\usepackage{subcaption}
\usepackage{diagbox}
\usepackage{longtable}
\usepackage{tcolorbox}
\usepackage[noabbrev,nameinlink]{cleveref}
\usepackage{bookmark}

\theoremstyle{thmstyleone}

\theoremstyle{thmstyletwo}
\newtheorem{remark}{Remark}
\newtheorem{example}{Example}
\newtheorem{claim}{Claim}

\theoremstyle{thmstylethree}
\newtheorem{axiom}{Axiom}
\newtheorem{definition}{Definition}
\newtheorem{theorem}{Theorem}
\newtheorem{lemma}{Lemma}

\Crefname{axiom}{axiom}{axioms}
\Crefname{axiom}{Axiom}{Axioms}
\Crefname{equation}{function}{functions}
\Crefname{equation}{Function}{Functions}

\newenvironment{longfigure}[2][]{%
  \expandafter\def\csname @captype\endcsname{figure}%
  \setcounter{subfigure}{0}%
  \setlength{\LTleft}{0pt}%
  #1%
  \begin{longtable}{#2}
}{\end{longtable}}

\newcommand{\defprob}[3]{
\begin{tcolorbox}[colback=gray!5!white,colframe=gray!75!black]
  \begin{minipage}{0.96\textwidth}
  \begin{tabular*}{\textwidth}{@{\extracolsep{\fill}}lr} #1   \\ \end{tabular*}
  {\bf{Input:}} #2  \\
  {\bf{Question:}} #3
  \end{minipage}
  \end{tcolorbox}
}

\newcommand{\centered}[1]{
    \begin{tabular}{l}
        \parbox{2cm}{\centering #1}
    \end{tabular}
}

\begin{document}

\title[Query-Based Committee Selection]{Query-Based Committee Selection}


\author[1]{\fnm{Itay Asher} \sur{Zimet} \orcid{https://orcid.org/0009-0008-5899-2846}}
\email{itay.zimet@gmail.com}

\author[3]{\fnm{Shiri} \sur{Alouf-Heffetz} \orcid{https://orcid.org/0000-0002-5102-2467}}
\email{shirihe@post.bgu.ac.il} 

\author*[2,3]{\fnm{Nimrod} \sur{Talmon} \orcid{https://orcid.org/0000-0001-7916-0979}}
\email{talmonn@bgu.ac.il}

\affil[1]{\orgdiv{The Future Scientists Center–Alpha Program}, \orgname{Davidson Institute, the Educational Arm of Weizmann Institute of Science}, \orgaddress{\city{Rehovot}, \country{Israel}}}

\affil[2]{\orgdiv{Department of Computer Science}, \orgname{Weizmann Institute of Science}, \orgaddress{\city{Rehovot}, \country{Israel}}}

\affil[3]{\orgdiv{Department of Industrial Engineering and Management }, \orgname{Ben-Gurion University of the Negev}, \orgaddress{\city{Beersheba}, \country{Israel}}}




\abstract{
\textbf{Purpose:} Multiwinner voting rules typically require full knowledge of voter preferences, which becomes impractical in large-scale or attention-limited settings. This paper investigates how accurately a winning committee can be approximated when voter preferences are elicited using a limited budget of structured queries.

\textbf{Methods:} We introduce a query-based framework for multiwinner elections in which voter preferences are elicited through refinement queries over subsets of candidates under a limited budget. We analyse several cost functions that model the cognitive effort needed to answer such queries, propose axiomatic properties for evaluating them, and experimentally evaluate simple query-based committee selection rules across multiple election models.

\textbf{Results:} Experimental results show that strategies based on recursively splitting candidate sets provide the best trade-off between elicitation cost and committee accuracy. Across several statistical models, these strategies approximate the outcome of k-Borda elections significantly more efficiently than alternative query types.

\textbf{Conclusion:} The results demonstrate that well-designed query strategies can substantially reduce the amount of preference information required while still producing high-quality committee outcomes, suggesting that query-based elicitation is a promising approach for scalable multiwinner decision-making.
}

\keywords{Computational social choice, Multiwinner elections, Incomplete information, experimental voting}



\maketitle

\section{Introduction}

The process of selecting a winning committee from a group of alternatives plays a crucial role in many societal, organizational, and political systems where \textit{group decision-making} is required~\cite{faliszewski2017multiwinner}. Multiwinner voting methods, such as positional scoring or Condorcet methods, are frequently used to determine the composition of such committees. However, for these methods to produce committees with certain fairness properties, they typically require full knowledge of voter preferences, which becomes impractical, especially in large-scale scenarios with numerous voters and candidates.
This limitation is particularly acute in decentralised governance systems such as those used in blockchain protocols. For example, in Cardano's Voltaire governance framework~\cite{cardanoVoltaire}, a large population is expected to participate in treasury and protocol decisions, often with little capacity to engage deeply. In such settings, voters may only have limited attention or willingness to express detailed preferences, yet high-stakes decisions must still be made. Query-based elicitation may reduce the cognitive burden on participants while preserving the representativeness and quality of outcomes.

In response to this challenge, we operate under a paradigm that assumes voters can only provide partial information due to limited attention. This approach builds on existing work in attention-aware social choice~\cite{alouf2022should,armstrong2024optimizing,talmon2023social,halpern2023representation,lang2020collective,imber2022approval,lu2020preference}.
In particular, we consider a model of committee selection based on a query-based framework (expanding on a similar framework for approval elections~\cite{halpern2023representation}). In our model, a \textit{fixed budget} is used to collect partial preferences from voters, where each query incurs a cost representing the effort required to answer it.
A central modelling challenge is therefore how to determine the cost each query entails: different representations of cognitive burden may lead to fundamentally different querying strategies.

The key problem then has two main aspects: how to quantify the cost of queries in a principled way, and how to elicit just enough information regarding a voter's preferences to determine a representative committee without incurring excessive costs. By framing committee selection as an optimisation problem of budget-constrained preference elicitation, we explore strategies for balancing the quality of the selected committee against the efficiency of the elicitation process.
We implement and experimentally evaluate different query-based rules under several electoral conditions, providing insights into the trade-offs between elicitation cost and committee quality. Our experiments, which include visualisations of election models using statistical frameworks, illustrate the effectiveness of the proposed methods in capturing voter preferences efficiently while achieving desirable committee outcomes.

\section{Related Work}

To contextualise our contributions, we survey the three bodies of literature most relevant to this work: multiwinner voting rules, partial preference elicitation, and query-based approaches to social choice.

\paragraph{Multiwinner Voting.}
The field of \textit{computational social choice} studies models of voting systems in which agents provide their preferences regarding various candidates. Prior research has examined election rules such as positional scoring, Condorcet, and approval voting that are used to determine the winners of an election~\cite{brandt2016handbook,brams1978approval}. In addition, multiwinner elections---in which a committee of a predetermined size is to be elected---have been studied extensively by both extending single-winner rules and proposing new election rules such as STV, Chamberlin-Courant, and Monroe~\cite{faliszewski2017multiwinner,elkind2017properties}.

\paragraph{Preference Elicitation Under Incomplete Information.}
These traditional methods often assume full access to voter preferences, which becomes infeasible in large-scale elections due to associated time and resource constraints~\cite{conitzer2007eliciting,distortion-Rosenchein}. Therefore, researchers have investigated partial preference elicitation methods, where only subsets of voter preferences are queried to infer broader preferences~\cite{ayadi2022approximating,uncertainty,bentert2020comparing,halpern2023representation}. These methods aim to collect just enough information to effectively determine the winning candidates, balancing the efficiency of data collection and the accuracy of the results.

\paragraph{Query-Based Elicitation.}
While querying voters for information was studied in the approval election setting~\cite{halpern2023representation}, ordinal preferences present unique challenges: voter preferences are richer and more descriptive, which makes them more challenging to elicit fully. Furthermore, the cost of a query in the ordinal setting depends on the cognitive effort it imposes, which is harder to quantify than in the binary approval case. We address this gap by modelling queries in the ordinal setting and proposing an axiomatic framework for evaluating their cost. We then explore different querying strategies and use rich simulations to study the trade-off between elicitation cost and the precision of the resulting committee, using the \textit{map of elections} library (\textit{Mapof} in short)~\cite{mapof,mapof_compass} to distinguish the quality of different strategies across statistical cultures and election types~\cite{gelman2002mathematics,lewis2001estimating,mattei2011empirical}.

\section{Preliminaries}\label{section:preliminaries}

We first discuss multiwinner elections. Specifically, we describe the format of the queries we consider in this paper and the cost associated with asking those queries in different settings. Then, we formally describe the \textit{query selection problem}.
We use $[n]$ to denote $\{1, \ldots, n\}$.

\subsection{Multiwinner Elections}

We introduce the standard model of multiwinner elections for the ordinal setting~\cite{faliszewski2017multiwinner}. 
An election consists of the following:

\begin{itemize}

\item
A set of candidates $C = \{c_1, \ldots, c_m\}$.

\item
A set of agents $V = \{1, \ldots, n\}$ where agent $i$ has a \emph{secret preference} $v_i$, which is a linear order (i.e., a permutation) of $C$.

\item
Target committee size $k \in [m]$.

\end{itemize}

Given an election $E = (C, V, k)$, a \textit{multiwinner voting rule} $\mathcal{R}$ returns a committee $\mathcal{R}(E)$, which is a subset of $C$ of size exactly $k$.

\subsection{Refinement Queries}

In our setting, we do not assume that we have the \textit{secret preferences} of the voters; on the contrary, we assume that initially we know nothing about voter preferences, but that we can query voters to get more information.
Next, we describe the type of queries that we consider in this paper.

\begin{definition}
Given a set $C$ of candidates, a \emph{refinement query} is of the form $Q = (C', B)$ where:
\begin{itemize}
    \item $C' \subseteq C$
    \item $B = (b_1, \ldots, b_\ell)\in{[0,1]}^\ell$ for some $1\leq \ell\leq |C'|$ such that $\sum_{j \in [\ell]} b_j = 1$
\end{itemize}
Then, the output of asking agent $i$ the query $Q$ is a weak order $R\subseteq C'\times C'$ that satisfies:
\begin{enumerate}
    \item  The weak order $R$ is consistent with the strict order $R'$ imposed by $v_i$. i.e., for all $x,y\in C'$:
    \begin{equation*}
        x\succ_{R}y\Rightarrow x\succ_{R'}y
    \end{equation*} 
    Equivalently, $R \subseteq R'\cap(C'\times C')$.
    
    \item The indifference classes $C'_1\succ_{R}C'_2\succ_{R}\dots\succ_{R}C'_\ell$ of $R$ are an ordered partition of $C'$ i.e.,
    \begin{equation*}
     C'_j\neq \varnothing,\quad \bigcup_{j\in[\ell]} C'_j = C',\quad C'_j\cap C'_k=\varnothing \quad\textrm{for}\ j\neq k
    \end{equation*}
    
    \item The indifference classes respect the proportions imposed by $B$. That is,
    \begin{equation*}
        \left||C'_j|-b_j|C'|\right|<1 \quad\textrm{for all}\quad j\in [\ell]
    \end{equation*}
    Intuitively, each class $C'_j$ has size as close to $b_j|C'|$ as integer rounding permits.
      
\end{enumerate}

\end{definition}

\begin{remark}
Note that our queries are general and expressive; they subsume many common elicitation formats. For example, a \emph{top-$t$ truncated ballot} over $C' \subseteq C$ corresponds to the query $(C', (\frac{1}{|C'|}, \ldots, \frac{1}{|C'|}, \frac{|C'|-t}{|C'|}))$ with $t+1$ buckets, where each of the first $t$ buckets isolates a single candidate and the last bucket collects the remaining $|C'|-t$ candidates. Similarly, a simple \emph{approval-style subset query}~\cite{halpern2023representation}---asking a voter to partition $C'$ into approved and disapproved candidates---is the special case $(C', (\frac{t}{|C'|}, \frac{|C'|-t}{|C'|}))$ for some threshold~$t$.
\end{remark}

\begin{example}\label{ex:querynext}
    Consider a voter 
    \begin{equation*}
        v_1=(a\succ b\succ c\succ d)
    \end{equation*}
    and the query    
    \begin{equation*}
        Q = (\{a, b, c, d\}, (\frac{1}{4},\frac{3}{4}))
    \end{equation*}
    Because the query partitions the candidates into two groups of relative sizes $\frac{1}{4}$ and $\frac{3}{4}$ and $|C'|=4$, the voter must separate exactly one favourite candidate from the remaining three.
    Hence, the query effectively asks: Of all these candidates, which is your favorite?
    As such, the voter's response to the query is $(\{a\}\succ_1 \{b,c,d\})$.

\end{example}

\begin{example}\label{ex:queryhalf}
    Consider the same voter $v=(a\succ b\succ c\succ d)$ and the query
    \begin{equation*}
        Q=(\{a, b, c, d\}, (\frac{1}{2},\frac{1}{2}))
    \end{equation*}
    Since $|C'|=4$, the query requires two indifference classes of size $2$.
    Therefore, the query corresponds to asking the voter to partition the candidates into a preferred half and a non-preferred half. As such, the response from the voter would be $(\{a,b\}\succ \{c,d\})$.
\end{example}

\subsection{Cost Functions}

We are interested in algorithms that minimize the amount of attention they demand from voters. Since different queries naturally impose varying levels of effort or cognitive load—and since this load may depend on characteristics of the voters—we introduce the notion of \textit{cost functions} to formally capture and quantify this variability. As a result, the choice of how to model query cost can greatly influence which elicitation strategies are considered efficient.

\begin{definition}
    Given the space of all possible queries $\Omega$ a cost function is a mapping 
    \begin{equation*}
        \mathrm{cost}: \Omega\xrightarrow{}R^+
    \end{equation*}
    that takes a refinement query $Q=(C',B)\in \Omega$ and returns a non-negative numeric value $\mathrm{cost}(Q)$ that acts as the query's cost and is deducted from the budget.
\end{definition}

We treat the cost function as an essential parameter of our model. The analysis of cost functions, their properties, and a few concrete examples for them are discussed in \Cref{section:cost-modeling}.

\subsection{The Query Selection Problem}

In this paper, we address the query selection problem: informally, given a cost function that quantifies the difficulty for a voter to respond to a particular query, our objective is to design a process that includes the following components: (1) a cost-efficient \textbf{querying strategy} that selects which queries to ask; and (2) a \textbf{voting rule} applied to the resulting partial preference data. The goal is to approximate as closely as possible the outcome—specifically, the selected committee—that would have been obtained if the complete preference information were available.

Slightly more formally, we are looking for \textit{query-based committee selection rules} that receive an \emph{election} $E = (C, V, k)$, a cost function, and a budget $b$ and proceed in two phases:
\begin{itemize}

\item
\textbf{Phase 1}:
The algorithm asks queries of total cost at most $b$ --- for the given cost function.

\item
\textbf{Phase 2}:
The algorithm selects a committee of size $k$ from $C$, aiming to approximate $\mathcal{R}(E)$.

\end{itemize}

The problem can be phrased as follows (note that the multiwinner voting rule used on the results of the queries has to be able to handle incomplete preferences).
\defprob{\textsc{The query selection problem}}{An election $E=(C,V,k)$, a cost function $\mathrm{cost}$, a budget $0 \leq b \leq \infty$, and a multiwinner voting rule $\mathcal{R}$}{How to best approximate $\mathcal{R}(E)$ by using refinement queries of cost at most~$b$ and using some multiwinner voting rule on the results of the queries?}

\section{Choosing Cost Functions}\label{section:cost-modeling}

In the previous section, we treated the cost function as a parameter of the query selection problem. We now investigate how to model query cost in a principled manner and select an appropriate cost function for a given circumstance.

\subsection{Basic Cost Functions}

We first explore a few intuitive cost functions. In particular, we initially investigate the following cost functions: 

One natural cost function is defined based on the number of candidates involved in a given query. This mapping aligns directly with the notion of query cost in the approval voting setting. Under this approach, queries that involve a larger number of candidates incur a higher cost. This idea can be formalised using the following expression: \begin{equation} \label{func:candidates}
    \mathrm{cost}(Q) = |C'|\ 
\end{equation}

\begin{example}
    Consider the queries $Q_1=(\{a,b,c,d\}, (\frac{1}{4},\frac{3}{4}))$ and $Q_2=(\{a,b,c,d\}, (\frac{1}{2},\frac{1}{2}))$ from \Cref{ex:querynext,ex:queryhalf}, respectively. Under this cost function, their costs are $\mathrm{cost}(Q_1)=\mathrm{cost}(Q_2)=|\{a,b,c,d\}|=4$. However, this fails to reflect that $Q_1$ should be cognitively less demanding than $Q_2$, since answering $Q_2$ requires the voter to select two preferred alternatives, while answering $Q_1$ requires selecting only a single preferred alternative.
\end{example}

To address this issue, observe that out of a query $ Q=(C',B)$ only the first $|B|-1$ classes require explicit selection by the voter (as the last class collects the remaining candidates), and thus the size of the last class should not be included in the cost calculation. The following expression captures this intuition: 
\begin{equation} \label{func:last bucket}
    \mathrm{cost}(Q) = |C'|\cdot(1-b_{|B|})\
\end{equation}
\begin{example}
    Now, the cost of the queries from the previous example would be:
    $\mathrm{cost}(Q_1)=|C'|(1-b_{|B|})=4(1-\frac{3}{4})=4\cdot\frac{1}{4}=1$ and $\mathrm{cost}(Q_2)=4(1-\frac{1}{2})=4\cdot \frac{1}{2}=2$. Nevertheless, if we compare $\mathrm{cost}(Q_2)=2$ to $\mathrm{cost}(Q_3)=(\{a,b,c,d\}, (\frac{1}{4},\frac{1}{4},\frac{1}{2}))=2$, we still fail to capture the inherent hardness of splitting candidates into more classes.
\end{example}

To improve the previous formulation, we require that a query with more classes will have a higher cost because we are asking a more granular question. This is represented in the following.
\begin{equation} \label{func:bucket_monotonicity}
    \mathrm{cost}(Q) = {(1-b_{|B|})\cdot |C'|}\cdot (|B| - 1)\
\end{equation}
With that said, there might be circumstances in which asking a question with more classes will have the same cost. For example, if the voters are computerised agents, then the number of classes might be a less pronounced factor affecting the query's cost.

Additionally, we might want to ensure that a query with more evenly distributed class sizes will cost more as it requires finer distinctions. We formalise this notion and some more properties to help guide the choice of cost functions in the next subsection.

\subsection{Properties of Cost Functions}\label{subsection:axiomatization}
We wish to develop a set of tools to evaluate a given cost function. To this end, we define some desired properties of cost functions in order to formalise the intuitions from the previous subsection and ensure that the cost functions we use can plausibly capture the cognitive effort imposed by a query.

\paragraph{Empty Questions}
First, we require that asking a question that does not give us information should not reduce our budget, as we can just as well not ask it.
\begin{axiom}[\emph{Empty questions}]\label{ax:empty_questions}
    For each $Q$ where $|B| = 0, 1$ or $|C'| = 0, 1$, $\mathrm{cost}(Q) = 0$.
\end{axiom}
This axiom can be satisfied by defining the cost function in a piecewise manner. Therefore, in what follows, we assume that all cost functions satisfy it.

\paragraph{Prefix Monotonicity}
Additionally, we require that if a query is an extended version of another query, then it should cost more, as it is asking the voter for more information.
\begin{definition} \label{def:prefix}
    We say that a vector $u$ is a nontrivial prefix of another vector $v$ if and only if $|u|<|v|$ and for every $i\in [|u|]$, $u_i=v_i$.
\end{definition}
\begin{axiom}[\emph{Prefix monotonicity}]\label{ax:prefix_monotonicity}
    For each $Q=(C',B), |B| \ne 0$ and $Q'=(C'',B'), |B'| \ne 1$, if $|C'|\cdot B$ is a nontrivial prefix of $|C''|\cdot B'$, then $\mathrm{cost}(Q) < \mathrm{cost}(Q')$.
\end{axiom}
Note that here we use the notation $|C'|\cdot B$ to describe the vector obtained by scaling each component of $B$ by $|C'|$.

\paragraph{Multiple Monotonicity}
This additionally requires that by multiplying the number of candidates included in the query, the cost scales by at least the same factor. This axiom provides a lower bound on the cost of queries that elicit more information from the voter.
\begin{definition} \label{def:query_mult}
    We say that a query $Q' = (C'', B)$ is a multiple of the query $Q=(C',B)$ by the scalar $l\in\mathbb{N}$ if and only if $|C''| = l\cdot |C'|$.
\end{definition}

\begin{axiom}[\emph{Multiple monotonicity}]\label{ax:query_mult_mono}
    For each $Q, Q'$, and $l\in \mathbb{N}$ such that $Q'$ is the multiple of $Q$ by the scalar $l$, $\mathrm{cost}(Q') \geq l\cdot \mathrm{cost}(Q)$.
\end{axiom}

\paragraph{Variance Monotonicity}
Now we attempt to formalise that the more \textit{extreme} a question is, the less it should cost. Hence, a question asking for the best candidate from a set $C$ will cost less than a question asking to split $C$ into two equal groups where one of them is preferred.
\begin{example}
    Let $C=C'=\{a,b,c,d,e,f,g,h,i,j\}$; we require that a question with $B=(1/2,1/2)$ will have a higher cost than a question with $B=(1/10,9/10)$.
\end{example}
\begin{definition} 
\label{def:variance}
    We define a measure for variance in the bucket ratios as:
    \begin{equation*}
        \mathrm{Var}(B) = \begin{cases} 
      0 & |B| = 0,1 \\
      \frac{1}{|B|}\cdot \sum_{j\in [|B|]}{(b_j-\frac{1}{|B|})}^2&|B|>1 \\
   \end{cases}
\end{equation*}
\end{definition}

\begin{axiom}[\emph{Variance monotonicity}]\label{ax:variance_mono}
    For each $Q=(C', B)$ and $Q' = (C',B')$ such that $|B|=|B'|$, if $\mathrm{Var}(B) > \mathrm{Var}(B')$ then $\mathrm{cost}(Q) < \mathrm{cost}(Q')$.
\end{axiom}

\subsection{Evaluation of Different Cost Functions}

We now examine each of the properties detailed above and prove which functions satisfy them.
In particular, we consider 
  \Cref{func:candidates,func:last bucket,func:bucket_monotonicity}.

\begin{claim} All the functions outlined above satisfy \label{claim:prefix_mono}\nameref{ax:prefix_monotonicity}. \end{claim}
\begin{proof}
We know that $|C'|\cdot B$ is a prefix of $|C''|\cdot B'$, therefore $|B|<|B'|$ and $|C'|<|C''|$. As such, \Cref{func:candidates} satisfies the axiom.

We now show that if $|C'|\cdot B$ is a prefix of $|C''|\cdot B'$ then $b'_{|B|}\leq 1-\frac{|C'|}{|C''|}$. Firstly, 
\begin{align*}
    |C''|\cdot b'_{|B'|}\leq \sum_{j=|B|+1}^{|B'|}|C''|\cdot b'_j= \sum_{j=1}^{|B'|}|C''|\cdot b'_j - \sum_{j=1}^{|B|}|C''|\cdot b'_j \\
    =\sum_{j=1}^{|B'|}|C''|\cdot b'_j - \sum_{j=1}^{|B|}|C'|\cdot b_j= |C''|\cdot\sum_{j=1}^{|B'|}b'_j - |C'|\cdot\sum_{j=1}^{|B|}b_j= |C''|-|C'|
\end{align*}
and because of that, $|C''|\cdot b'_{|B'|}\leq |C''|-|C'|$ and in conclusion: $b'_{|B'|}\leq 1-\frac{|C'|}{|C''|}$.
Lastly, we can conclude that in the case of \Cref{func:last bucket}: $\mathrm{cost}(Q')=|C''|\cdot (1-b'_{|B'|})\geq |C''|\cdot (1-(1-|\frac{|C'|}{|C''|}))=|C'|>|C'|\cdot (1-b_{|B|}) = \mathrm{cost}(Q)$.
Now we know that $|C''|\cdot (1-b'_{|B'|})> |C'|\cdot (1-b_{|B|})$. Combining that with the knowledge that $|B'|> |B|$ and thus $|B'|-1> |B|-1$ we can conclude that $|C''|\cdot (1-b'_{|B'|})\cdot(|B'|-1)\geq |C'|\cdot (1-b_{|B|})\cdot(|B|-1)$ and thus $\mathrm{cost}(Q') > \mathrm{cost}(Q)$ and the axiom is satisfied for \Cref{func:bucket_monotonicity}.
\end{proof}
\begin{claim} All the functions outlined above satisfy \nameref{ax:query_mult_mono}. \end{claim}
\begin{proof}
    All three functions mentioned above satisfy this property, as all of them scale with $|C'|$.
\end{proof}
\begin{claim} 
None of the functions outlined above satisfies \nameref{ax:variance_mono}.
\end{claim}
\begin{proof}
    For \Cref{func:candidates} we can see that $|C'|=|C'|$ and, as such, $\mathrm{cost}(Q)=|C'|=\mathrm{cost}(Q')$ so the function does not satisfy the property.

Let $B=[0.5,0.3,0.2]$ and $B'=[0.4,0.35,0.25]$. A calculation gives:
\begin{align*}
    \mathrm{Var}(B)&=\frac{7}{450}\\
    \mathrm{Var}(B')&=\frac{7}{1800}
\end{align*}
As such, $\mathrm{Var}(B)> \mathrm{Var}(B')$. Let $|C'|=10$, for \Cref{func:last bucket}: 
\begin{align*}
    \mathrm{cost}(Q)&=10\cdot (1-0.2)= 8\\
    \mathrm{cost}(Q')&=10\cdot (1-0.25)=7.5
\end{align*}
Therefore, $\mathrm{cost}(Q)>\mathrm{cost}(Q')$ and \nameref{ax:variance_mono} is not satisfied. 
For \Cref{func:bucket_monotonicity}:
\begin{align*}
    \mathrm{cost}(Q)&=16\\
    \mathrm{cost}(Q')&=15
\end{align*}
Therefore, $\mathrm{cost}(Q)>\mathrm{cost}(Q')$ and \nameref{ax:variance_mono} is not satisfied again.
\end{proof}

\subsection{An Additional Cost Function}

In the preceding section we explored a variety of cost functions and their associated properties; in this section, we propose an additional cost function, for which we show that it satisfies all axioms, including \nameref{ax:variance_mono}.

\begin{theorem} The function 
\begin{equation} \label{func:perfect}
    \mathrm{cost}(Q)=|C'|\cdot |B|\cdot (1-\mathrm{Var}(B))
\end{equation}
satisfies all the axioms mentioned above.
\end{theorem}
\begin{proof}

Let $Q=(C',B)$ and $Q'=(C'',B')$ be such that $|C'|\cdot B$ is a prefix of $|C''|\cdot B'$. It follows that $\mathrm{cost}(Q)=|C'|\cdot |B|\cdot (1-\mathrm{Var}(B))$. 
\begin{lemma} \label{lm:dVar}
    We now show that the defined variance function is equivalent to a standard variance function on a specific discrete random variable. Let $B^*$ be a discrete random variable that can take on any value in $B$ with probability proportional to the number of times it appears in $B$; first, we note that the average value of $B^*$ is $\frac{1}{|B|}$ because it has $|B|$ values and sums to 1. Now we conclude that $\mathrm{Var}(B^*)=\frac{1}{|B|}\cdot\sum_{i=1}^{|B|}{(b_i-\frac{1}{|B|})^2}=\mathrm{Var}(B)$.
\end{lemma}
\begin{lemma} \label{lm:VarBound}
    Let $B$ be a vector with positive entries that sum to 1. According to \Cref{lm:dVar} $\mathrm{Var}(B)$ can be dealt with as a variance on a random variable and as such according to the Bhatia--Davis inequality~\cite{bhatia}:
    \begin{equation*}
        \mathrm{Var}(B)\leq (MAX(B)-\frac{1}{|B|})\cdot (\frac{1}{|B|} - MIN(B))\leq (1-\frac{1}{|B|})\cdot\frac{1}{|B|}
    \end{equation*}
    Therefore: $1-\mathrm{Var}(B)\geq \frac{|B|^2-|B|+1}{|B|^2}$.
\end{lemma}

We now show that $\mathrm{cost}(Q') > \mathrm{cost}(Q)$. By \Cref{lm:VarBound}:
\begin{align*}
    \mathrm{cost}(Q') &= |C''|\cdot |B'|\cdot(1-\mathrm{Var}(B'))\\
    &\geq |C''|\cdot |B'|\cdot \frac{|B'|^2-|B'|+1}{|B'|^2} \tag{a}\label{ineq:varbound}\\
    &= |C''|\cdot \frac{|B'|^2-|B'|+1}{|B'|} \\
    & =|C''|\cdot|B'|-1+\frac{1}{|B'|}\\
    &\geq |C''|\cdot \left(|B|+\frac{1}{|B'|}\right) \tag{b}\label{ineq:Bprime}\\
    &\geq |C'|\cdot |B|+\frac{|C'|}{|B'|}\\
    &> |C'|\cdot |B|\\
    &\geq |C'|\cdot |B|\cdot (1-\mathrm{Var}(B)) = \mathrm{cost}(Q) \tag{c}\label{ineq:varpos}
\end{align*}
where~\eqref{ineq:varbound} follows from \Cref{lm:VarBound},~\eqref{ineq:Bprime} uses the fact that $|B'|\geq |B|+1$ (since $|C'|\cdot B$ is a nontrivial prefix of $|C''|\cdot B'$), and~\eqref{ineq:varpos} holds because $0 \leq \mathrm{Var}(B) \leq 1$.
Thus, $\mathrm{cost}(Q) < \mathrm{cost}(Q')$ and \nameref{ax:prefix_monotonicity} is satisfied.

Let $Q=(C',B), Q'=(C'',B)$ and $l \in \mathbb{N}$ be such that $Q'$ is the multiple of $Q$ by $l$. Then a simple computation gives:  \begin{equation*}\mathrm{cost}(Q')=l\cdot |C'|\cdot |B|\cdot(1-\mathrm{Var}(B))=l\cdot \mathrm{cost}(Q)\ ,\end{equation*} and \nameref{ax:query_mult_mono} is satisfied.

Let $Q=(C', B)$ and $Q' = (C',B')$ be such that $|B|=|B'|$ and $\mathrm{Var}(B) > \mathrm{Var}(B')$. It directly follows that $1-\mathrm{Var}(B)<1-\mathrm{Var}(B')$, which means that $\mathrm{cost}(Q)=|C'|\cdot |B|\cdot (1-\mathrm{Var}(B))< |C'|\cdot |B|\cdot (1-\mathrm{Var}(B'))=|C'|\cdot |B'|\cdot (1-\mathrm{Var}(B'))=\mathrm{cost}(Q')$ therefore, $\mathrm{cost}(Q)\leq \mathrm{cost}(Q')$ and \nameref{ax:variance_mono} is satisfied. 

\end{proof}

\subsection{A Computational Cost Function}
Up until now, we looked almost exclusively at cost functions and properties of them from the human-oriented perspective. Now, we present a grounded setting where the cost function follows naturally and rigorously.

\begin{example}
    Consider a Euclidean election between computerised agents, where each candidate and voter is represented by a point in d-dimensional space. The preferences of a voter are then decided by their distance to each candidate. Thus, the full preference order of a voter is the list of candidates sorted by their distance from the voter.
    
    Then, asking a query $Q=(C',B)$ is equivalent to computing all distances from the voter to the candidates in $|C'|$, selecting $|B|$ $k$th order statistics of the distances array and partitioning said array based on these statistics. This can be done recursively by partitioning the array around the middle statistic each time, amounting to $\log |B|$ recursive levels that cost $|C'|$ each, yielding an asymptotic runtime of $O(|C'|\log |B|)$.

    Thus, the most natural cost function for this use case is:
    \begin{equation}\label{func:computer}
        \mathrm{cost}(Q)=|C'|\log |B|
    \end{equation}
\end{example}

\begin{claim}
    \Cref{func:computer} satisfies \nameref{ax:empty_questions}, \nameref{ax:prefix_monotonicity} and \nameref{ax:query_mult_mono} but does not satisfy \nameref{ax:variance_mono}.
\end{claim}
\begin{proof}
    First, we established previously that \nameref{ax:empty_questions} is regarded as part of a piecewise definition for all cost functions and thus is satisfied. 
    Now, as we noted in the proof of \Cref{claim:prefix_mono}, if $|C'|\cdot B$ is a prefix of $|C''|\cdot B'$, then $|B|<|B'|$ and $|C'|<|C''|$. Therefore:
    \begin{equation*}
        |C'|\log |B|<|C''|\log |B'|
    \end{equation*}
     and \nameref{ax:prefix_monotonicity} is satisfied.
    Additionally, \Cref{func:computer} scales with the factor $|C'|$ and thus satisfies \nameref{ax:query_mult_mono}.

    Finally, let
    \begin{align*}
        Q&=(C'=\{a,b,c,d\},B=(\frac{1}{4},\frac{3}{4}))\\ Q&=(C'=\{a,b,c,d\},B'=(\frac{1}{2},\frac{1}{2}))
    \end{align*}
    A calculation provides: $\mathrm{Var}(B)=\frac{1}{16}>\mathrm{Var}(B')=0$ but $|B|=|B'|=2$ and thus 
    \begin{equation*}
        \mathrm{cost}(Q)=|C'|\log |B|=|C'|\log |B'|=\mathrm{cost}(Q')
    \end{equation*}
    Therefore \nameref{ax:variance_mono} is violated.
\end{proof}
    
\begin{table}[t]
\centering
\renewcommand{\arraystretch}{1.5}
\setlength{\tabcolsep}{0.5\tabcolsep}
\caption{Different cost functions and their properties.}
\begin{tabular}{c|c|c|c}
\centered{\diagbox{Function}{Axiom}}
& \nameref{ax:prefix_monotonicity}
& \nameref{ax:query_mult_mono}
& \nameref{ax:variance_mono} \\
\midrule
$\lvert C' \rvert$
& YES & YES & NO \\
$\lvert C' \rvert \cdot (1 - b_{\lvert B \rvert})$
& YES & YES & NO \\
$(1 - b_{\lvert B \rvert}) \cdot \lvert C' \rvert \cdot (\lvert B \rvert - 1)$
& YES & YES & NO \\
$\lvert C' \rvert \cdot \lvert B \rvert \cdot (1 - \operatorname{\mathrm{Var}}(B))$
& YES & YES & YES \\
$\lvert C' \rvert\log \lvert B \rvert$
& YES & YES & NO \\
\end{tabular}
\end{table}

\section{Possible Query-Based Rules}\label{section:Possible-rules}

Each \emph{query-based rule} is made up of the two phases mentioned above, where the first phase is the data collection phase in which the rule strategically queries the voters, and the second is the data processing stage where the rule uses the data gathered to find a winning committee. 
We present and evaluate a few rudimentary and non-adaptive rules that try to approximate the output of a given positional scoring rule (like $k$-Borda). 
This allows us to begin exploring the space of all possible query-based rules.
We start with some basic observations.

\subsection{Phase 1 -- the Data Collection Phase}
 When looking at the data collection phase, we can argue that there is no apparent benefit to asking a voter about candidates that we have already determined are in different classes. Slightly more formally, if we have two classes $C_1 \subseteq C, C_2 \subseteq C$ such that $ C_1 \cap C_2 =\varnothing $ and we know the relation between them for the given voter, then there is no benefit to asking about $C_3$ if there are candidates $x\in C_1,\ y\in C_2$ such that $\{x,y\}\subseteq C_3$.
 
 This motivates a basic approach for the first phase that iteratively splits the voter preferences into classes and then continues by splitting those classes recursively, giving a partial order on the candidates.

 In such an approach, each rule is distinguished from the others by the queries asked and the way they are formed. Additionally, the rules are differentiated by the way that they use the budget and distribute it among the voters.
 
 \subsection{Phase 2 -- the Data Processing Phase}
After the querying phase, each voter's partial information can be represented as an ordered partition $C'_1 \succ C'_2 \succ \cdots \succ C'_\ell$ of the candidate set $C$, where candidates within the same class are indistinguishable. Given a positional scoring vector $\mathbf{s} = (s_1, \ldots, s_m)$ (such as the Borda vector $s_i = m - i$), we assign each candidate $c \in C'_j$ the score:
\begin{equation}\label{eq:partial-score}
    \text{score}(c) = \frac{1}{|C'_j|} \sum_{r = L_j}^{L_j + |C'_j| - 1} s_r
\end{equation}
where $L_j = 1 + \sum_{i=1}^{j-1} |C'_i|$ is the first position occupied by class $C'_j$. That is, each candidate in a class receives the average of the positional scores over all positions that the class spans.

\begin{example}
    Let $C = \{a,b,c,d\}$ with Borda scores $\mathbf{s} = (3,2,1,0)$, and suppose the querying phase yields the partial order $\{a,b\} \succ \{c,d\}$. Then candidates $a$ and $b$ each receive $\frac{1}{2}(s_1 + s_2) = \frac{1}{2}(3+2) = 2.5$, and candidates $c$ and $d$ each receive $\frac{1}{2}(s_3 + s_4) = \frac{1}{2}(1+0) = 0.5$. If no queries are asked (i.e., the partition is the trivial single class $\{a,b,c,d\}$), every candidate receives $\frac{1}{4}(3+2+1+0) = 1.5$.
\end{example}

The total score of a candidate is the sum of its scores across all voters, and the winning committee consists of the $k$ candidates with the highest total scores.\footnote{We note that extending this approach beyond positional scoring rules (e.g., to committee scoring rules such as Chamberlin--Courant) requires additional modelling choices, which we leave to future work.}

 \subsection{Algorithms}
 
We present several simple rules that stem from the approaches for the different phases that are detailed above, and explore them experimentally. To that end, we introduce several question types (intuitive classes of queries), and two intuitive strategies for distributing the budget among the voters.

We have the following question types:
\begin{itemize}
    \item \textit{\textbf{N}ext}: the best-ranked candidate out of a set of candidates $C'$.
    \item \textit{\textbf{L}ast}: the worst-ranked candidate out of a set $C'$.
    \item \textit{\textbf{N}ext and \textbf{L}ast}: the best-ranked and the worst-ranked candidate out of a set $C'$.
    \item \textit{\textbf{S}plit}: divide the set $C'$ into two approximately equally-sized sets $A_1$ and $A_2$ (even/odd) such that any candidate in $A_1$ is preferred over candidates in $A_2$ and $C'=A_1\cup A_2, A_1\cap A_2 = \varnothing$.
\end{itemize}

Additionally, we have the following two ways of spending the budget: 
\begin{itemize}
    \item \textbf{Eq}ually: most intuitively, we split the budget almost equally among the voters. We ask queries iteratively from one voter to the next one in a predefined order. 
    \item \textbf{F}irst-\textbf{c}ome, \textbf{f}irst-\textbf{s}erved: we first use all budget necessary to get the full order of the first voter, then, and only then we move on to the second voter.
\end{itemize}

\begin{table}[t]
     \centering
     \renewcommand{\arraystretch}{1.5}
     \caption{Definitions of the different strategies considered.}
     \begin{tabularx}{\textwidth}{c|>{\raggedright\arraybackslash}X|>{\raggedright\arraybackslash}X}
          \diagbox{question type}{budget distribution} & divides budget \textbf{EQ}ually & divides greedily (\textbf{FCFS}) \\\midrule
         \textit{\textbf{N}ext} ($B=(1/|C|,1-1/|C|)$)& \emph{next equally (N-EQ)} & \emph{next FCFS (N-FCFS)} \\
          \textit{\textbf{L}ast} ($B=(1-1/|C|,1/|C|)$)& \emph{last equally (L-EQ)} & \emph{last FCFS (L-FCFS)}\\
          \textit{\textbf{N}ext and \textbf{L}ast} ($B=(1/|C|,1-2/|C|,1/|C|)$)& \emph{next and last equally (NL-EQ)} & \emph{next and last FCFS (NL-FCFS)}\\
          \textit{\textbf{S}plit} ($B=(1/2,1/2)$)& \emph{split equally (S-EQ)} & \emph{split FCFS (S-FCFS)}\\
     \end{tabularx}
      \label{tab:defenitions}
 \end{table}

The question types and the budget strategies together make up the rudimentary strategies we explored in our simulations as illustrated in \Cref{tab:defenitions}. From here onward we use the abbreviations denoted in \Cref{tab:defenitions} when referring to the different strategies.

\begin{example}
    Let $E=(C,V,k)$ such that:
    \[
        C=\{1,2,3,4\}\qquad
        V=\{v_1=(1\succ2\succ3\succ4),v_2=(2\succ1\succ3\succ4)\}\qquad
        k=2
    \]
    And let the budget $b=24$. Based on S-EQ, we start by asking $v_1$ a query $Q=(C,(0.5,0.5))$. Notice that $\mathrm{cost}(Q)=|C|\cdot|(0.5,0.5)|\cdot(1-\mathrm{Var}((0.5,0.5)))=4\cdot 2\cdot 1=8$ and so the budget left after asking this question is $24-8=16$. From that question, we learn that $v_1$ prefers $\{1,2\}$ to $\{3,4\}$. Then, splitting the budget equally, we ask $v_2$ the same question and get the same answer back, and we are left with a budget $b=8$. Similarly, we ask both voters about their preferences regarding $\{1,2\}$, costing us $4$ for each question and depleting the budget. Then we treat it as if $v_1=1\succ2\succ\{3,4\}$ and $v_2=2\succ1\succ\{3,4\}$. Running Borda we get $\{1,2\}$ as our winning committee.
\end{example}

\section{Experiments}
We conducted three types of experiments.\footnote{Code for the experiments is available at: \url{https://github.com/itayzimet/VotingRulesExperiment}.}  
In each of our experiments we set a traditional multiwinner voting rule (such as $k$-Borda defined below) as the goal for the query-based rule and compared their outputs on different elections and with different budgets.

Each query-based rule was scored on a specific election based on the Hamming distance between the winning committee of the election as determined by the traditional rule and the winning committee produced by the query-based rule, i.e., the number of candidate swaps needed to get from one committee to the other. We use Hamming distance as it directly measures the number of ``incorrect'' committee members and is straightforward to interpret; note that for committees of fixed size $k$, the Hamming distance equals $2k(1 - J)$ where $J$ is the Jaccard similarity, so the choice between the two metrics is equivalent up to an affine transformation. In our experiments, the traditional rule used for comparison is k-Borda. For comparison, we also evaluated a voting rule that picks the winning committee randomly (shown as an additional baseline in \Cref{fig:map_by_random,fig:IC-results,fig:Euclidean-results}).

\begin{definition}[\textit{k}-Borda]
    In \textit{k}-Borda, each voter assigns a score to each candidate based on the total number of candidates $|C|$ and the candidate's position in the voter preference ranking. Specifically, the score assigned to the candidate ranked in position $i$ is given by:
    \[
    \text{score}(i) = |C| - i.
    \]
    The total Borda score of each candidate is obtained by summing their scores across all voters. The winning committee consists of the $k$ candidates with the highest total Borda scores.
\end{definition}

\begin{example}
    Consider an election with three candidates, $C = \{a, b, c\}$, and one voter whose preference ranking is $a \succ b \succ c$. The Borda scores assigned by this voter are computed as follows:
    \begin{itemize}
        \item Candidate $a$, ranked 1st: $\text{score}(1) = 3 - 1 = 2$
        \item Candidate $b$, ranked 2nd: $\text{score}(2) = 3 - 2 = 1$
        \item Candidate $c$, ranked 3rd: $\text{score}(3) = 3 - 3 = 0$
    \end{itemize}
    If $k = 2$, then the winning committee consists of candidates $a$ and $b$, as they have the highest Borda scores.
\end{example}

\subsection{Experimental Design}

We performed three experiments, as discussed below.

\subsubsection{Uniform and 2D Populations Experiments} \label{subsec-IC-experiments}
In the first two experiments, we considered two simple election structures:

\begin{definition}[Impartial culture]
     Under the Impartial Culture model (IC), every preference order appears with the  same probability. That is, to generate a vote we choose a preference order uniformly at random.
 \end{definition}

\begin{definition}[Euclidean culture]
    Under the Euclidean model, every candidate and voter is assigned a point in d-dimensional space randomly. Then the preference order of each voter is the list of candidates sorted by their distance to the voter, such that a closer candidate is more preferred. In our simulations the points were assigned uniformly from the two-dimensional square $[0,1]^2$.
\end{definition}
 
Specifically, we generated 5000 IC and 2D Euclidean elections with 100 candidates and 100 voters. Each election was given to every query-based rule and to a traditional multiwinner rule ($k$-Borda) with a committee size of 50. Then we scored the query-based rules on the different elections with different budget constraints. The scores were averaged over the different elections to produce \Cref{fig:IC-results,fig:Euclidean-results}. For the IC case we used \Cref{func:perfect} as the cost function and for the Euclidean case we used \Cref{func:computer}.
\subsubsection{Mapof Experiment}

In the third experiment, we used the Mapof framework~\cite{mapof,mapof_compass} to generate elections from different statistical models such as P{\'o}lya-Eggenberger Urn, Euclidean election, and Mallows. We then mapped the different elections in \Cref{fig:map_by_type} based on their similarities and embedded a ``compass'' for the mapping consisting of four elections that represent different extreme properties:

\begin{itemize}
     \item An identity election (ID) is an election where all voters agree on a single preference order, modelling perfect agreement.
     \item A uniformity election (UN) is an election where each voter takes each preference equally often, modelling perfect disagreement.
     \item A stratification election (ST) is an election where all the voters agree that half of the candidates are better than the other but do not agree on the order inside those halves, modelling partial agreement.
     \item An antagonism election (AN) is where half of the voters have preference orders exactly opposite the other half's preferences, modelling conflict.
 \end{itemize}

 We then generated a map (\Cref{fig:fig_table}) for each query-based rule where the colour of an election on the map is the sum of the scores for each budget value (each budget's score is averaged over five random voter orders) normalised by the highest score (the highest Hamming distance). For this experiment we used \Cref{func:perfect} as the cost function.

\section{Results and Discussion}

We now present and discuss the findings of the experiments that evaluated the query-based rules from \Cref{section:Possible-rules} on various voting models. We categorize the different query-based rules following the classification in \Cref{tab:defenitions}.

\subsection{Impact of the Query Choosing Strategy}
The best query-based rules for IC and two-dimensional Euclidean elections are those that split the candidates into two buckets as illustrated in \Cref{fig:IC-results,fig:Euclidean-results}. That is, their ratio of cost to distance from the true committee is almost always better than the other rules. It is also evident from \Cref{fig:IC-results,fig:Euclidean-results} that the differences between strategies can be substantial. For example, one strategy can achieve perfect accuracy with a budget of about 130,000 (S-EQ) while others need an order of magnitude greater budget to get perfect accuracy (NL-FCFS). Moreover, the best-performing rules in both cases are the same family of rules and their behaviour is almost identical.

\begin{figure}[t]
    \centering
    \begin{minipage}[c]{0.49\textwidth}
        \includegraphics[width=\textwidth]{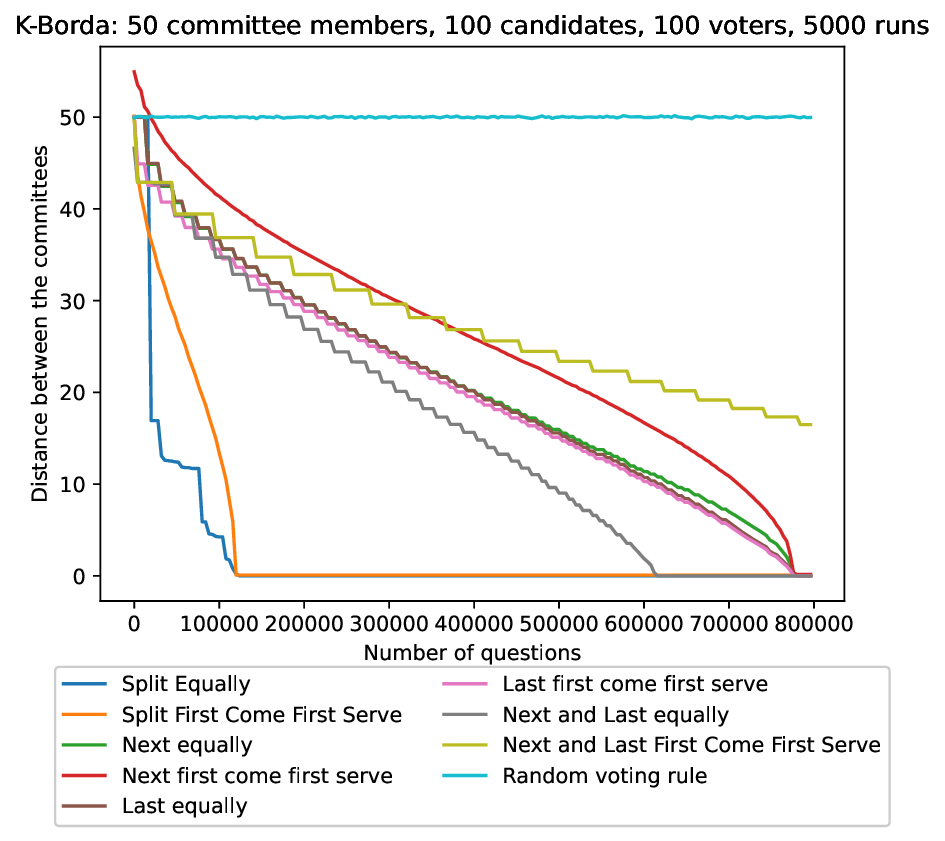}
    \caption{Committee distance from k-Borda winner under various voting rules and budget constraints (avg. over 5000 IC elections).}
    \label{fig:IC-results}
    \end{minipage}
    \begin{minipage}[c]{0.49\textwidth}
        \includegraphics[width=\textwidth]{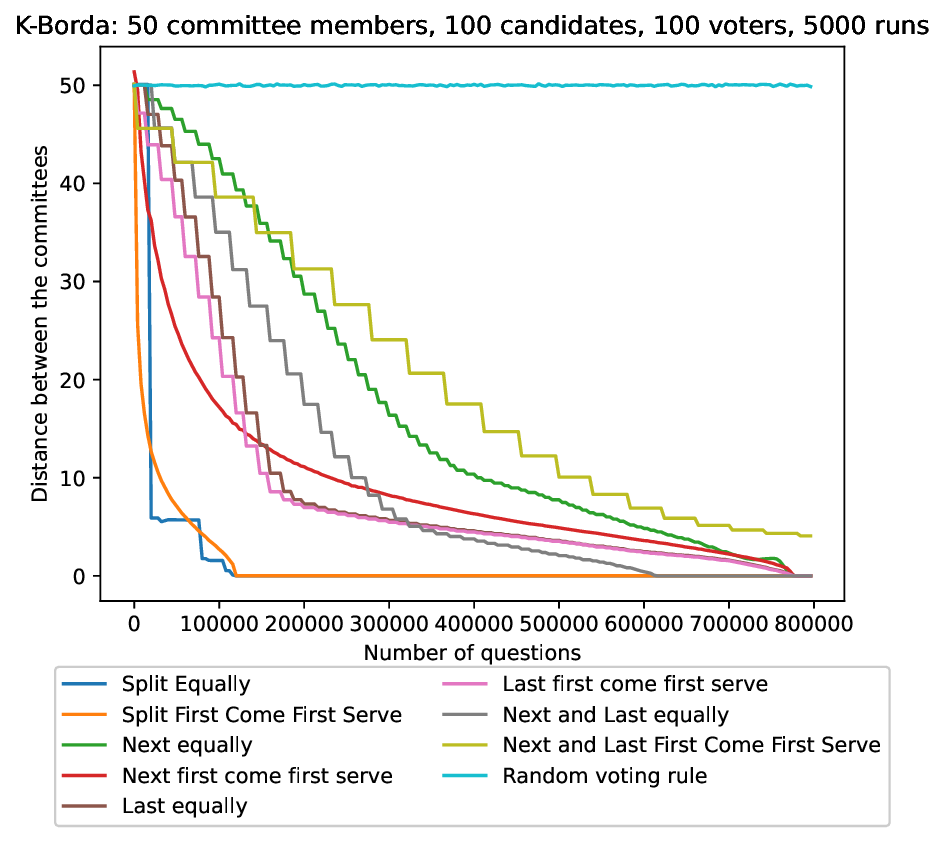}
    \caption{Committee distance from k-Borda winner under various voting rules and budget constraints (avg. over 5000 two-dimensional elections).}
    \label{fig:Euclidean-results}
    \end{minipage}
\end{figure}

Additionally, the conclusions from \Cref{fig:IC-results} are not limited only to IC elections. That is because if we look at the maps in \Cref{fig:fig_table} in relation to the culture map at \Cref{fig:map_by_type} we observe that, for most elections, S-FCFS and S-EQ (\Cref{fig:map_by_splitFCFS,fig:map_by_spliteq}) are the best strategies with respect to all the statistical models. As such, we conclude that with high probability taking the \textit{split strategy} (\Cref{fig:map_by_spliteq,fig:map_by_splitFCFS}) is the best choice (under \Cref{func:perfect}). Finally, we note that (\textit{next, last}) strategies almost always perform worse than the alternatives (under both \Cref{func:perfect,func:computer}). For reference, \Cref{fig:map_by_random} shows that the random baseline performs poorly uniformly across all election types, confirming that the improvements observed for the query-based strategies are meaningful rather than artifacts of the election structure.

\begin{figure}[t]
    \centering
    \begin{minipage}[c]{0.49\textwidth}
        \includegraphics[width=\textwidth]{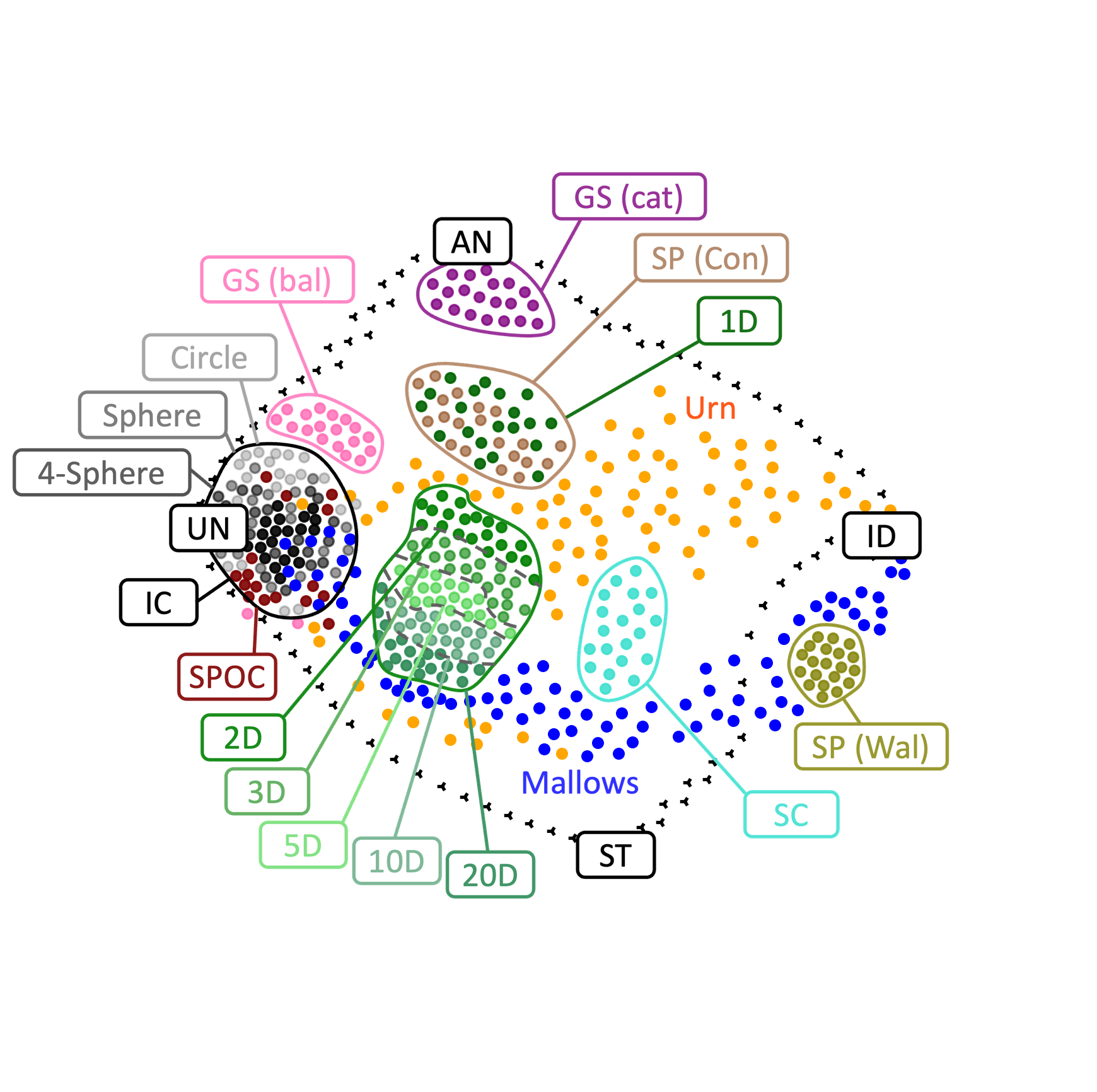}
    \caption{Map of elections coloured by election family.} \label{fig:map_by_type}
    \end{minipage}
    \begin{minipage}[c]{0.49\textwidth}
        \includegraphics[width=\textwidth]{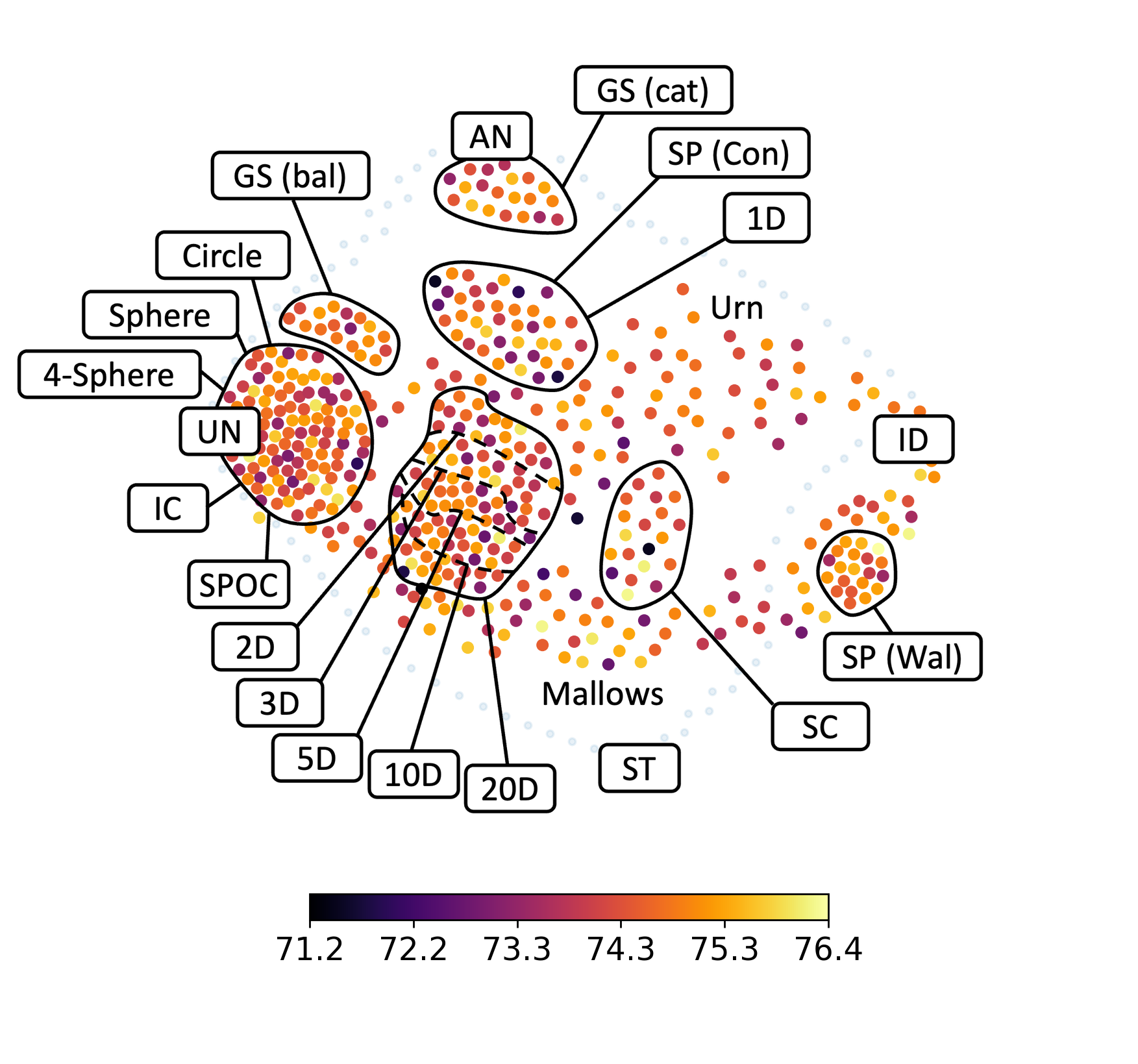}
    \caption{Map of elections coloured by performance of random.} \label{fig:map_by_random}
    \end{minipage}
\end{figure}
\begin{longfigure}{c|cc}
    \caption{Election maps coloured by difficulty for different query-based rules.}
    \label{fig:fig_table} \endfirsthead
    \caption*{Election maps coloured by difficulty for different query-based rules (continued).}
\endhead
  \addtocounter{figure}{-1}
\diagbox[width=0.28\textwidth]{question \\ type}{budget \\ distribution}& divides Equally (EQ) & divides greedily (FCFS)  \\\midrule
     \centered{split}&\begin{subfigure}{0.360\textwidth}
        \includegraphics[width=\textwidth]{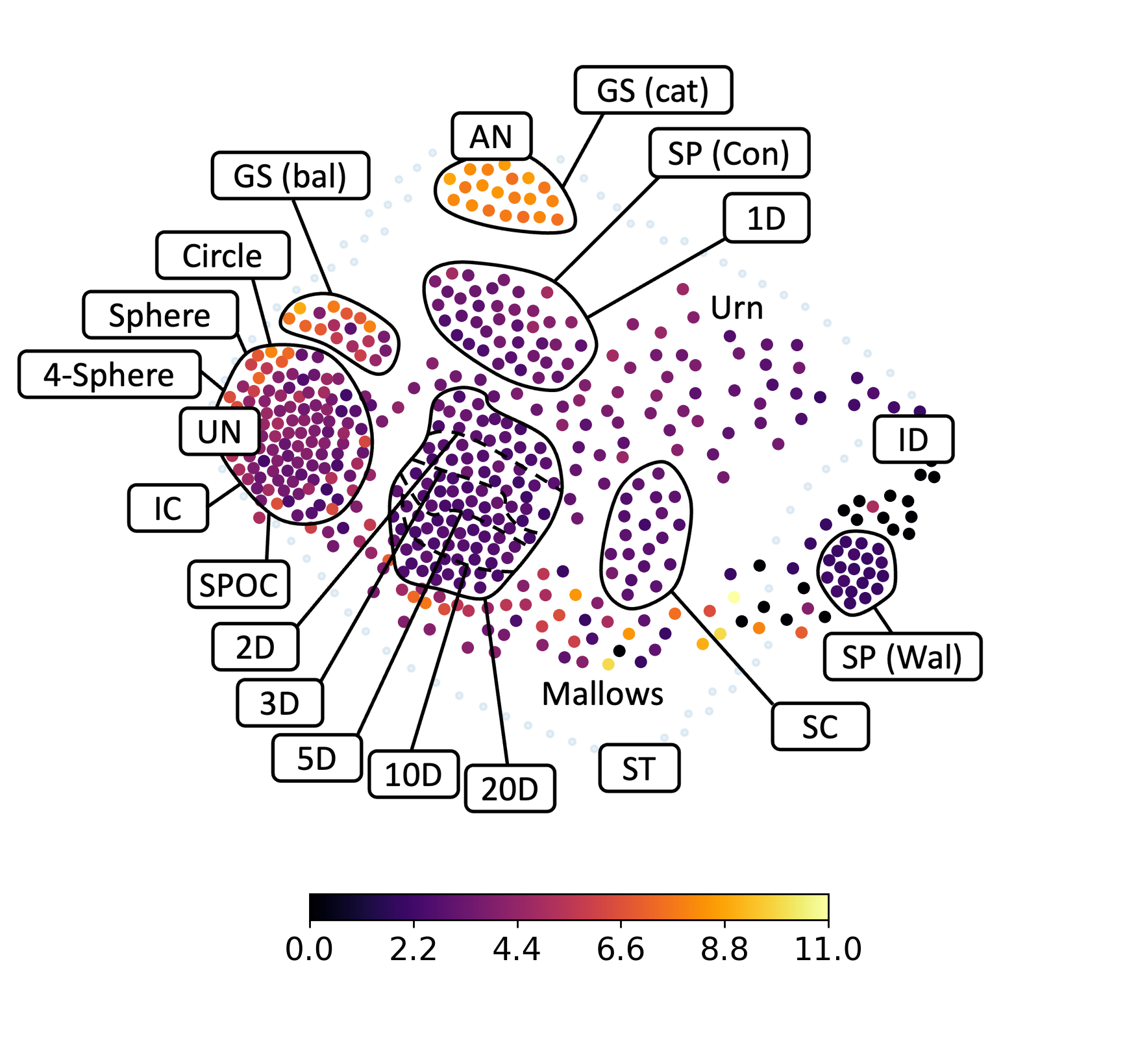}
        \caption{Map of elections coloured by difficulty for S-EQ} \label{fig:map_by_spliteq}
    \end{subfigure}&\begin{subfigure}{0.360\textwidth}
        \includegraphics[width=\textwidth]{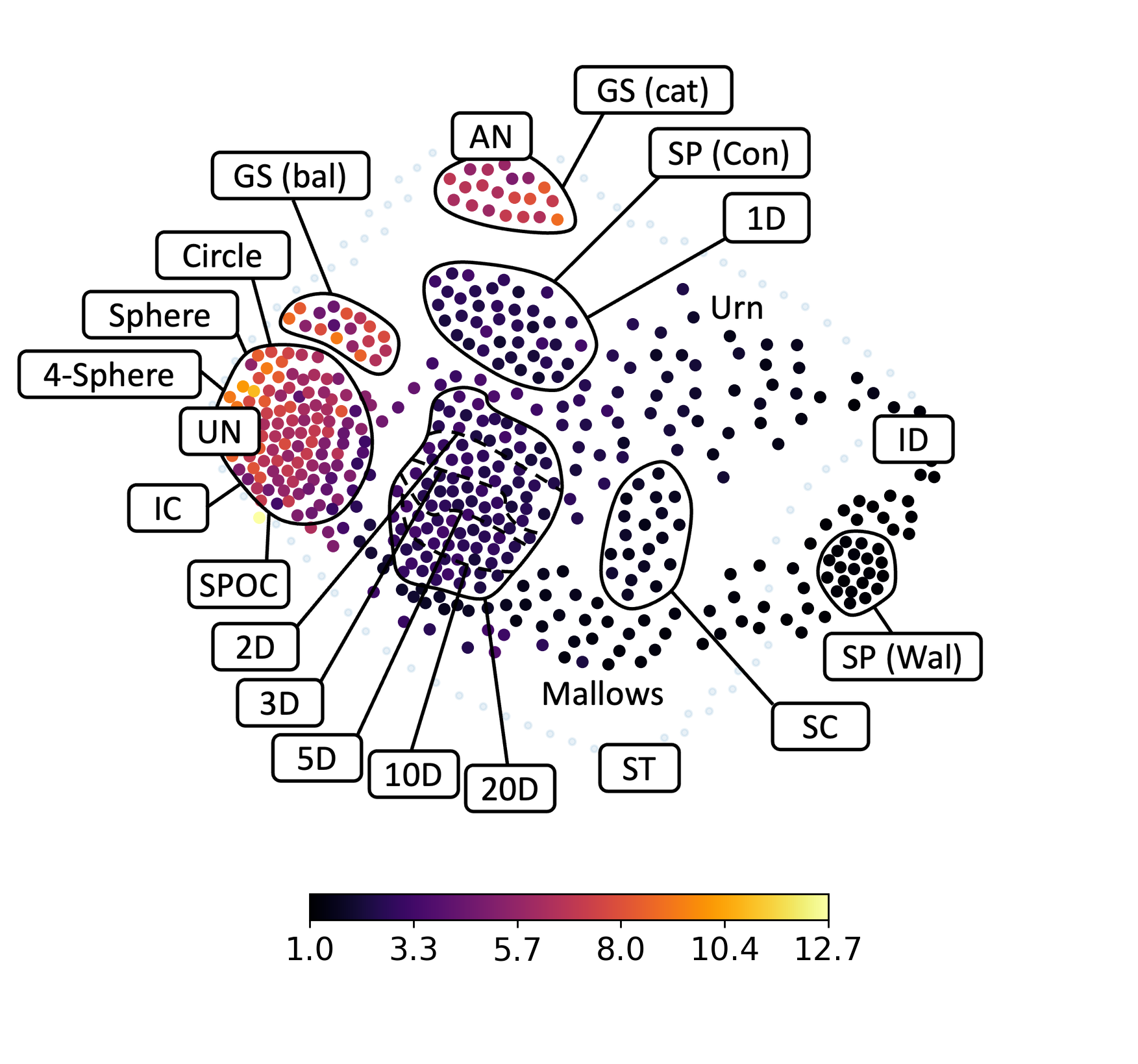}
        \caption{Map of elections coloured by difficulty for S-FCFS} \label{fig:map_by_splitFCFS}
    \end{subfigure}\\
     \centered{last}&\begin{subfigure}{0.360\textwidth}
        \includegraphics[width=\textwidth]{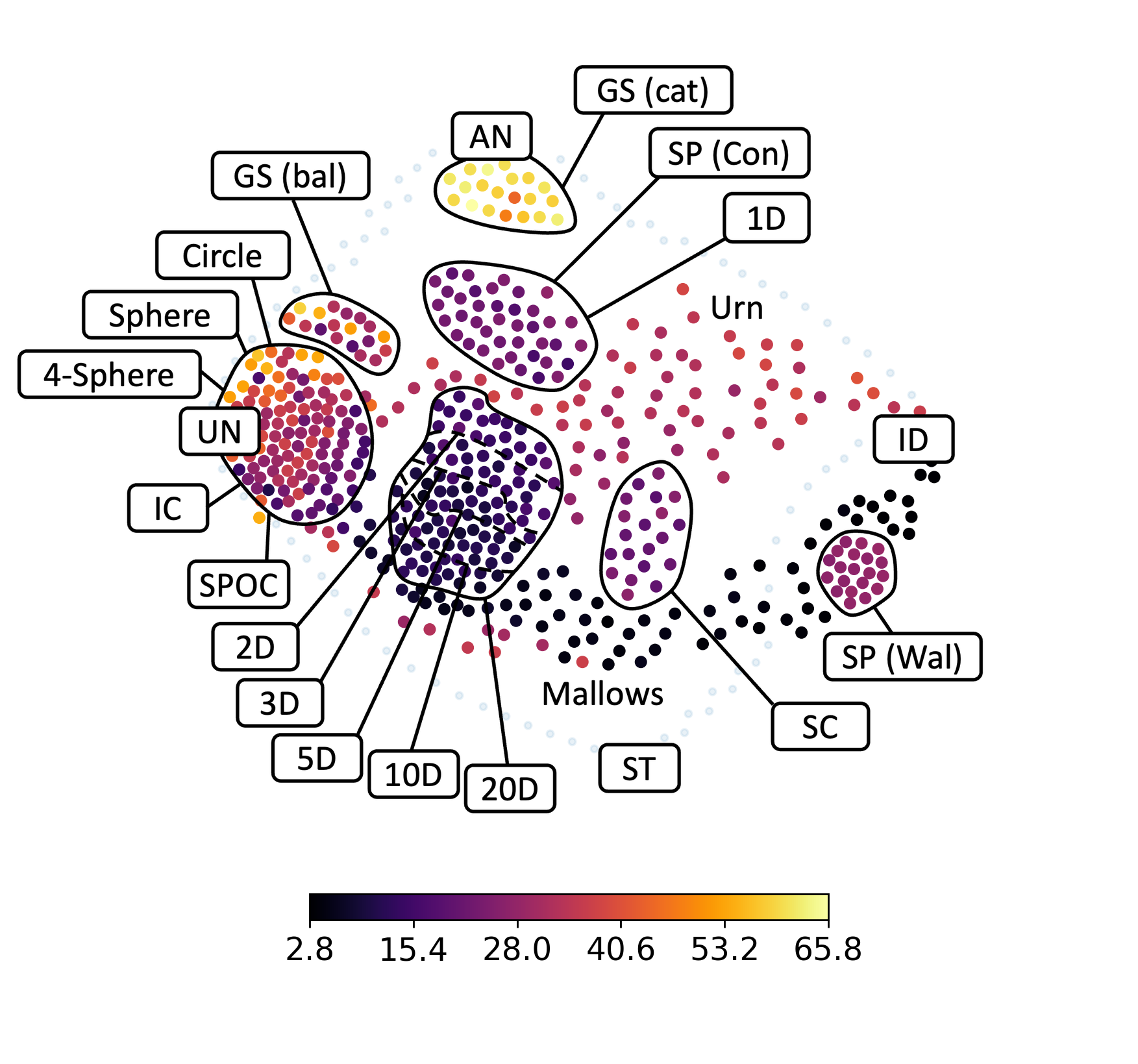}
        \caption{Map of elections coloured by difficulty for L-EQ} \label{fig:map_by_lasteq}
    \end{subfigure}&
    \begin{subfigure}{0.360\textwidth}
        \includegraphics[width=\textwidth]{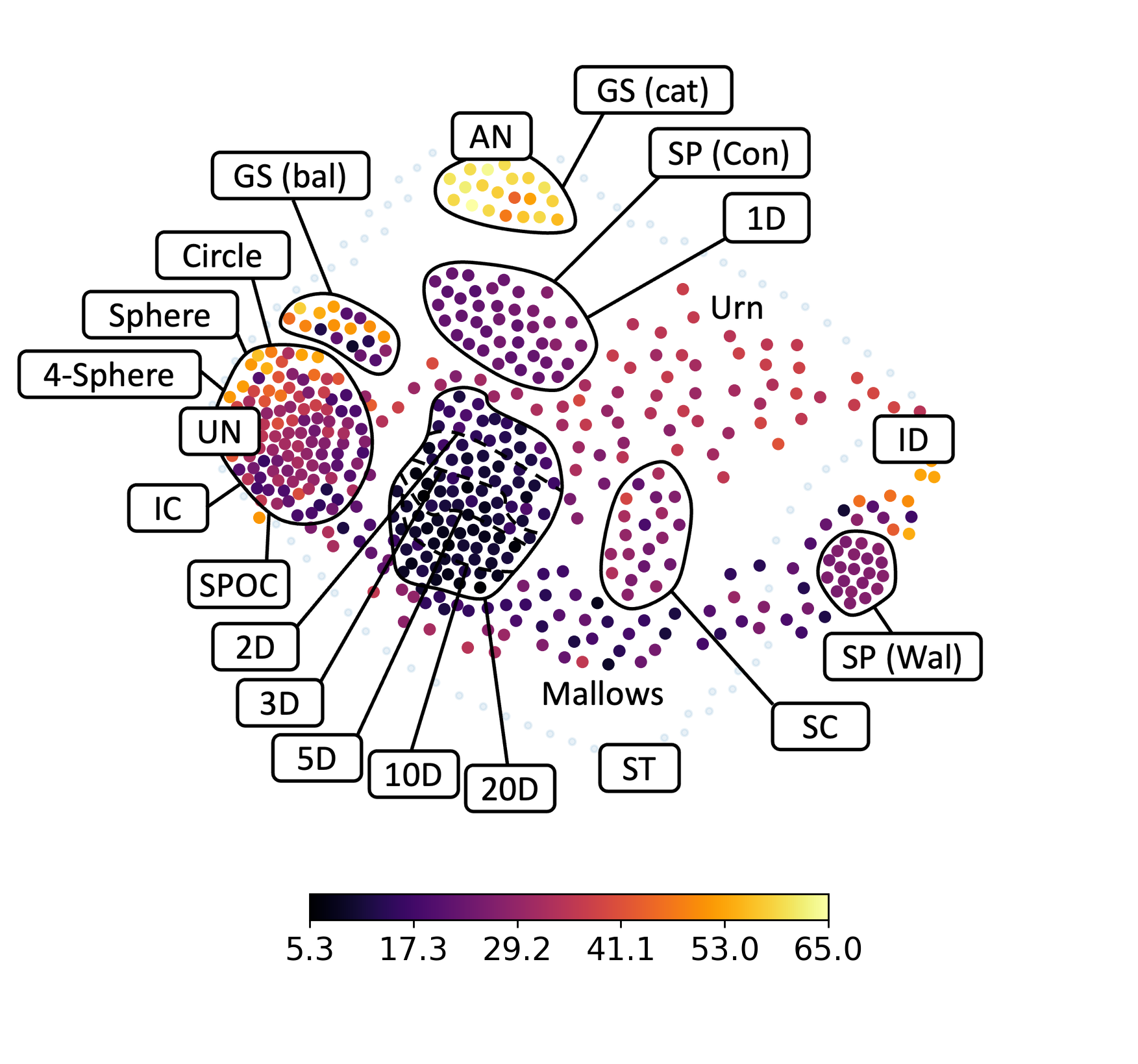}
        \caption{Map of elections coloured by difficulty for L-FCFS} \label{fig:map_by_lastFCFS}
    \end{subfigure} \\
    \centered{next}&
    \begin{subfigure}{0.360\textwidth}
        \includegraphics[width=\textwidth]{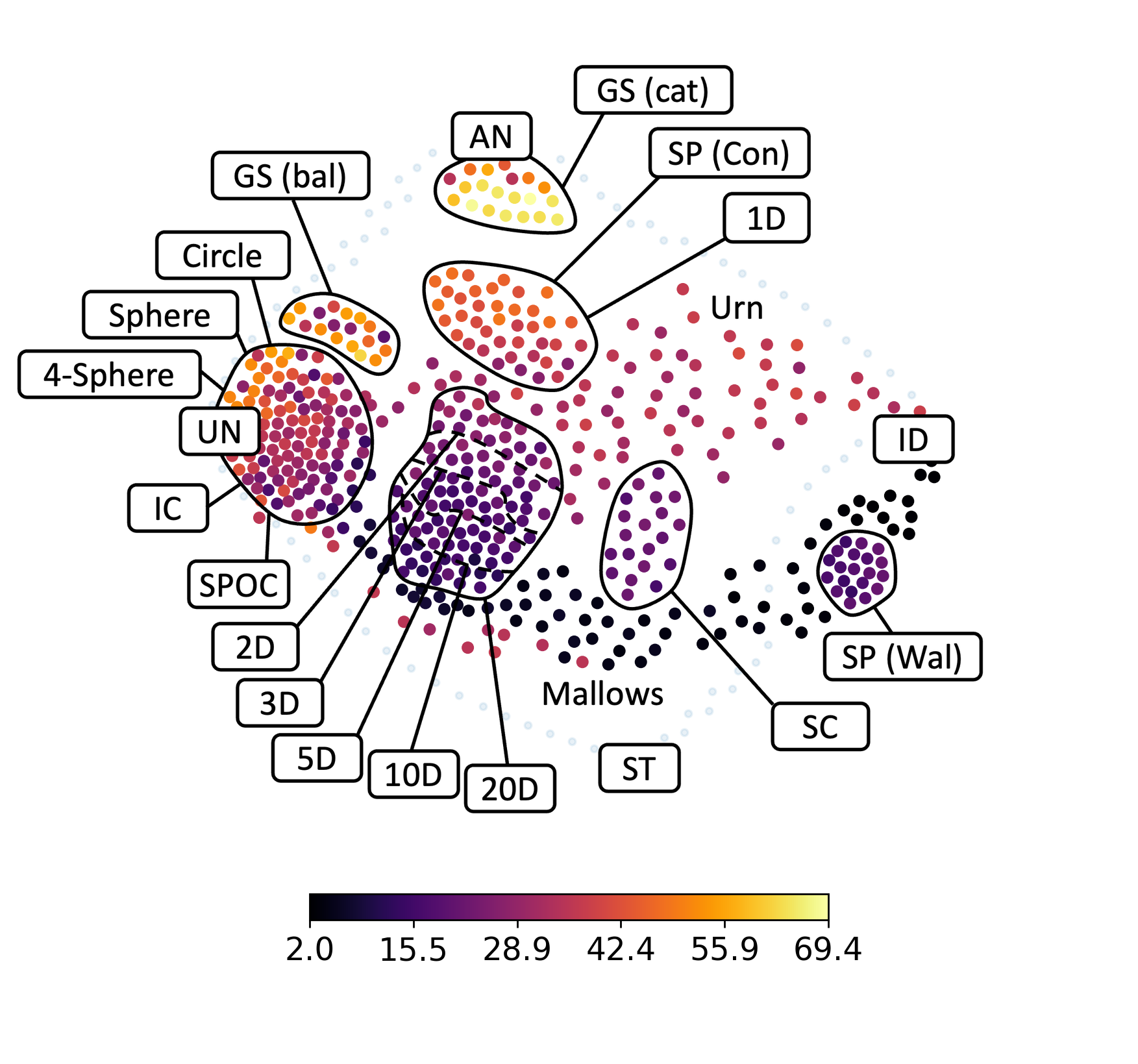}
        \caption{Map of elections coloured by difficulty for N-EQ} \label{fig:map_by_next_eq}
    \end{subfigure}&
    \begin{subfigure}{0.360\textwidth}
        \includegraphics[width=\textwidth]{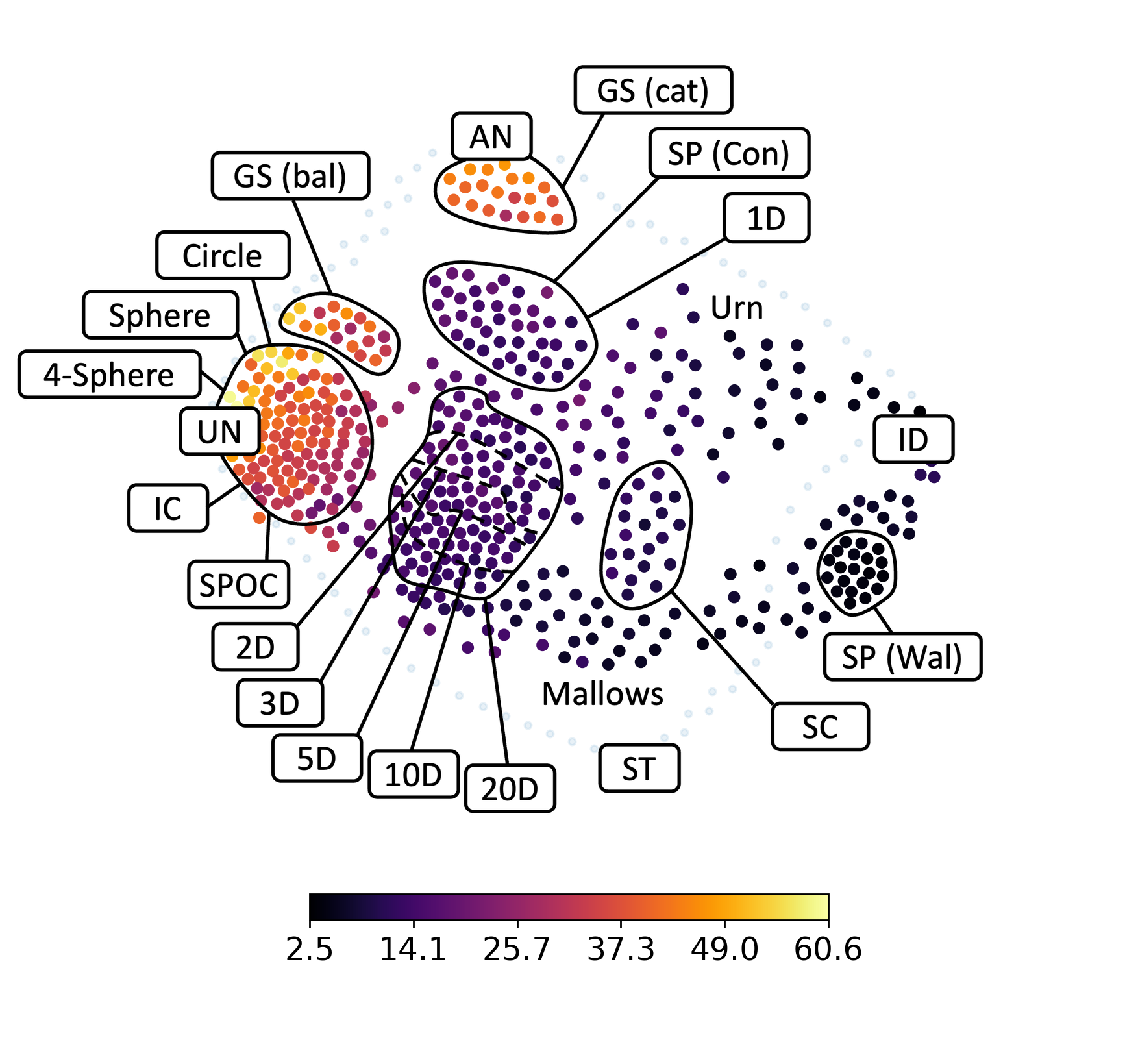}
        \caption{Map of elections coloured by difficulty for N-FCFS} \label{fig:map_by_next_FCFS}
    \end{subfigure} \\
    \centered{next and last}&
    \begin{subfigure}{0.360\textwidth}
        \includegraphics[width=\textwidth]{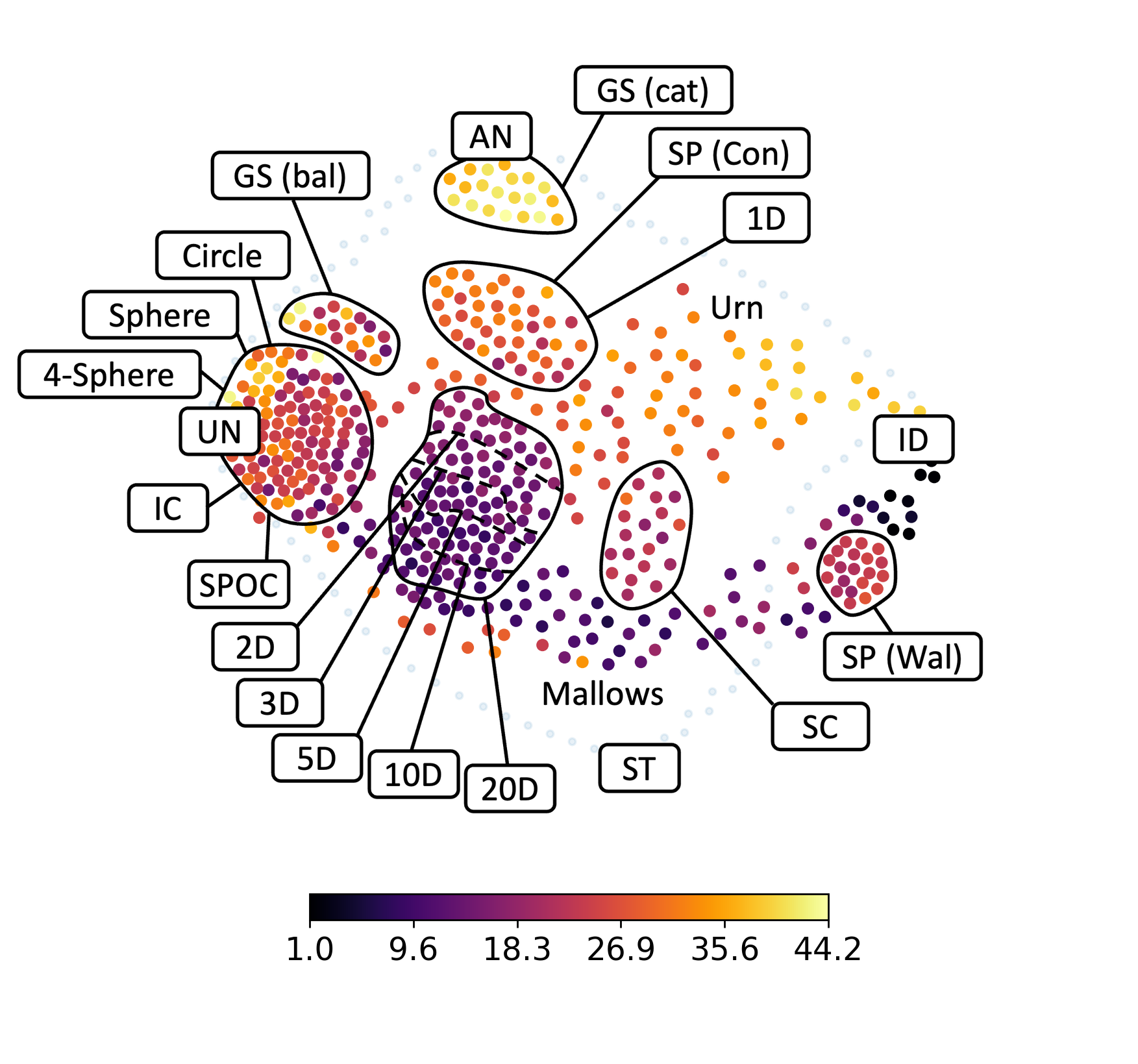}
        \caption{Map of elections coloured by difficulty for NL-EQ} \label{fig:map_by_next_lasteq}
    \end{subfigure} &
    \begin{subfigure}{0.360\textwidth}
        \includegraphics[width=\textwidth]{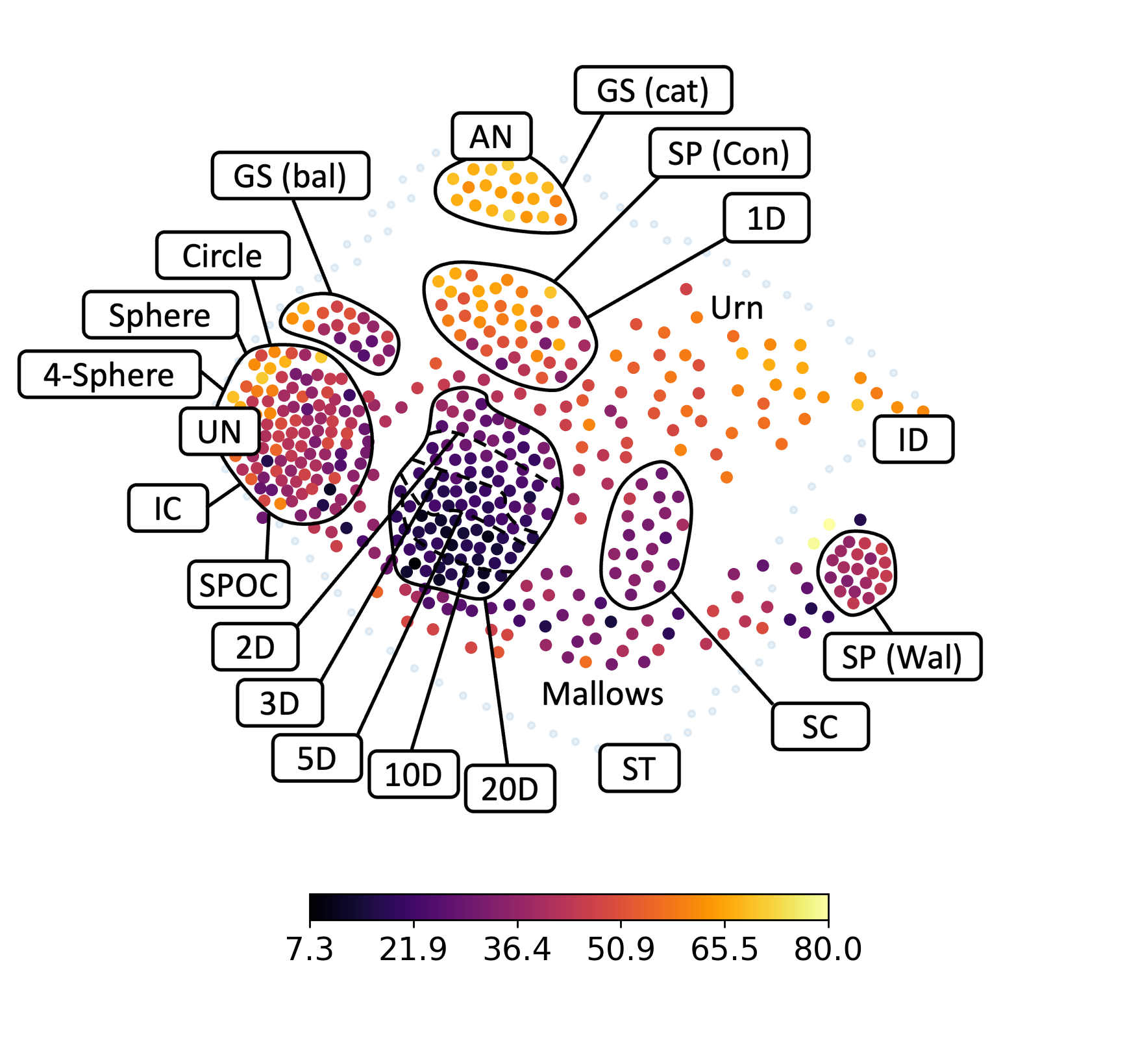}
        \caption{Map of elections coloured by difficulty for NL-FCFS} \label{fig:map_by_next_lastFCFS}
    \end{subfigure} 

\end{longfigure}

\subsection{Impact of the Budget Distribution}

First, we can see that in the IC and two-dimensional cases shown in \Cref{fig:IC-results,fig:Euclidean-results} the budget distribution is asymptotically negligible, though for the most part, splitting equally seems to improve the performance of the various strategies.

Second, we can see in the culture maps in \Cref{fig:fig_table} that in all the querying strategies tested, the FCFS variant performs better closer to ID and worse in elections closer to UN, though in some cases the difference is negligible. This makes sense as the less variation there is in public opinion, the more we gain by probing a single person thoroughly.

\subsection{Difficulty of Elections Closer to UN and AN}

The elections closer to UN and AN proved to be a greater challenge for every rule tested. This could indicate an inherent hardness for said elections compared to elections closer to ST and ID, or it could indicate that the rules tested here are not comprehensive enough. Either way, this indicates the need for further research on the topic.

\section{Conclusion}

We began this paper by asking how effectively we can query a community with limited resources to estimate the consensus of their preferences. 
Our findings suggest that, for both chosen cost functions, the best strategy is to use binary search questions to split the candidates into two groups --- one preferred by the voter over the other --- while allocating the budget equally among all the voters.

These findings are particularly relevant for large-scale decentralised governance settings, such as blockchain-based treasury systems (e.g., Cardano's Voltaire framework), where participants face severe attention constraints yet collective decisions must still reflect community preferences. Our results suggest that even simple binary-split queries, distributed equally across participants, can yield high-quality committees at a fraction of the elicitation cost required by full preference collection.

Furthermore, we observed that the statistical culture used to generate the voter preferences and the cost function used to evaluate queries had a substantial influence on the optimal strategy. 
In particular, in the case of our human-tailored cost function (\Cref{func:perfect}), when the preferences were closer to a uniform or antagonistic distribution (as in the uniform experiments in \Cref{fig:IC-results} and in elections close to UN and AN in \Cref{fig:fig_table}), it was advantageous to collect relatively coarse information regarding each voter's preferences, probably because the election results are more sensitive to small changes in the preferences of each voter as a clear consensus does not exist.
In contrast, in elections whose preference profiles were closer to an identity or a stratification distribution (elections close to ID and ST in \Cref{fig:fig_table}), it was more beneficial to acquire highly detailed information about the preferences of a single voter.
A likely explanation is that, in such elections, obtaining more detailed information about an individual voter is particularly valuable because this voter is representative of many other voters.
Furthermore, this outcome may depend on the specific multiwinner voting rule employed and the target committee size. For instance, a rule designed to shortlist a diverse set of candidates may yield superior performance in uniformly distributed elections, where the importance of broad consensus is comparatively diminished.

Additionally, elections closer to a uniform distribution appear to have an inherent ``hardness'' to them: regardless of the strategy chosen, they consistently yielded the worst scores. A likely explanation is that more dispersed voter profiles make consensus harder to reach, leaving the election results more sensitive to the details of each individual voter's preferences.

The results presented here establish a foundation for a broader research agenda.
The space of possible query-based voting rules is vast, and the rules and cost functions we considered represent a first step into that space.
Future work could explore more efficient querying rules that make use of online algorithms that adapt the query type and the budget usage dynamically --- allowing the use of multiple query types and potentially inferring voter preferences from societal trends. Additionally, more work could be done on improving the accuracy of the cost functions -- by proposing new axioms and conducting experimental research into the cognitive burdens imposed by queries.

\backmatter



\bmhead{Acknowledgements}
Author Talmon acknowledges support by the European Union under the Horizon Europe project \href{https://perycles-project.eu/}{Perycles} (Participatory Democracy that Scales). Views and opinions expressed are however those of the author(s) only and do not necessarily reflect those of the European Union or the European Research Executive Agency. Neither the European Union nor the granting authority can be held responsible for them.
\smallskip
\begin{center}
\includegraphics[width=0.5\textwidth]{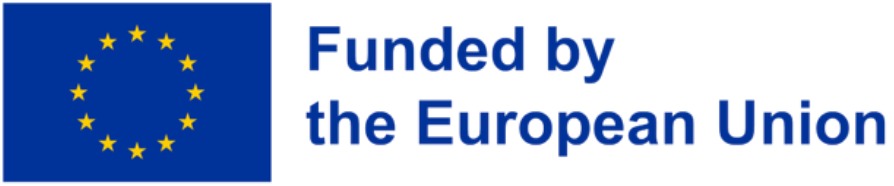}
\end{center}

The authors also gratefully acknowledge the support of the Alpha Program and The Future Scientists Center for the Advancement of the Talented and Gifted, funded by the Maimonides Fund, whose support was instrumental in making this paper possible.

\section*{Declarations}

\subsection*{Competing Interests}
The authors declare that they have no competing interests.

\subsection*{Funding Information}
This work was funded partially by the European Union under the PERYCLES project (grant agreement No.\ 101177658).

\subsection*{Author Contribution}
All authors contributed equally to all aspects of this work, including conceptualization, methodology, analysis, and writing.

\subsection*{Data Availability Statement}
Data and code are available from the corresponding author upon reasonable request.

\subsection*{Research Involving Human and Animals}
Not Applicable.

\subsection*{Informed Consent}
Not Applicable.

\begin{appendices}



\end{appendices}
\bibliography{sn-bibliography}
\end{document}